\documentclass[11pt]{article}

\usepackage[margin=1in]{geometry}

\usepackage{microtype}
\usepackage{graphicx}
\usepackage{subcaption}
\usepackage{booktabs} 
\usepackage{times}    

\usepackage[numbers]{natbib}
\usepackage{hyperref}
\usepackage{url}

\usepackage{algorithm}
\usepackage{algorithmic}
\usepackage{xcolor}

\usepackage{amsmath}
\usepackage{commath}
\usepackage{amssymb}
\usepackage{mathtools}
\usepackage{amsthm}

\usepackage{macros}

\theoremstyle{plain}
\newtheorem{theorem}{Theorem}[section]

\newtheorem{lemma}[theorem]{Lemma}

\theoremstyle{definition}
\newtheorem{definition}[theorem]{Definition}

\theoremstyle{remark}

\usepackage[textsize=tiny]{todonotes}

\usepackage{amsmath}
\usepackage{thm-restate}
\usepackage{cleveref} 
\usepackage{xspace}
\usepackage{dblfloatfix}
\usepackage[section]{placeins}

\newcommand{\rprune}{\textsc{RobustPrune}\xspace}
\newcommand{\urprune}{\textsc{UnsortedRobustPrune}\xspace}

\newcommand{\rpt}{\textsc{RP-Tuning}\xspace}

\newcommand{\recall}{$100$-Recall@$100$\xspace}

\usepackage{tikz}
\usetikzlibrary{arrows.meta,calc,positioning}

\tikzset{
  >=Latex,
  line cap=round,
  line join=round,
  every node/.style={font=\small},
  pt/.style={circle, draw=none, inner sep=0pt, minimum size=5.2pt},
  axis/.style={gray!65, dashed, line width=0.75pt},
  edge/.style={black, densely dotted, line width=0.95pt},
  dim/.style={<->, black, line width=0.85pt},
  dimlab/.style={font=\small, fill=white, inner sep=1.2pt},
  arcl/.style={gray!70, dashed, line width=0.9pt},
  dist/.style={black, line width=0.95pt},
  lab/.style={font=\small, fill=white, inner sep=1.2pt}
}

\begin{document}

\title{Prune, Don't Rebuild: Efficiently Tuning $\alpha$-Reachable Graphs for Nearest Neighbor Search}

\author{
  Tian Zhang \quad Ashwin Padaki \quad Jiaming Liang \quad Zack Ives \quad Erik Waingarten \\[6pt]
  University of Pennsylvania \\
  \texttt{\{tianzh, apadaki, liangjm, zives, ewaingar\}@seas.upenn.edu}
}

\date{}

\maketitle

\begin{abstract}

Vector similarity search is an essential primitive in modern AI and ML applications.
Most vector databases adopt graph-based approximate nearest neighbor (ANN) search algorithms, such as DiskANN~\cite{NEURIPS2019_09853c7f}, which have demonstrated state-of-the-art empirical performance. DiskANN's graph construction is governed by a reachability parameter $\alpha$, which gives a trade-off between construction time, query time, and accuracy. However, adaptively tuning this trade-off typically requires rebuilding the index for different $\alpha$ values, which is prohibitive at scale. In this work, we propose $\rpt$, an efficient post-hoc routine, based on DiskANN's pruning step, to adjust the $\alpha$ parameter without reconstructing the full index. Within the $\alpha$-reachability framework of prior theoretical works~\cite{indyk2023worstcaseperformancepopularapproximate, gollapudi2025sort}, we prove that pruning an initially $\alpha$-reachable graph with $\rpt$ preserves worst-case reachability guarantees in general metrics and improved guarantees in Euclidean metrics. Empirically, we show that $\rpt$ accelerates DiskANN tuning on four public datasets by up to 43$\times$ with negligible overhead.
 
\end{abstract}

\section{Introduction}
Vector similarity search is an essential primitive in modern AI and ML applications, such as retrieval augmented generation~\cite{gao_retrieval-augmented_2023}, linking and matching for knowledge graph construction~\cite{zeakis_pre-trained_2023}, and curating datasets for model training and fine-tuning.
Most vector database implementations adopt graph-based approximate-nearest-neighbor (ANN) search algorithms. At scale, DiskANN~\cite{NEURIPS2019_09853c7f} is widely used, e.g., in Milvus~\cite{wang_milvus_2021} and Microsoft's CosmosDB~\cite{upreti_cost-effective_2025}, both because of its state-of-the-art empirical performance and its ability to scale the index structure beyond RAM (in contrast to more widely known algorithms like HNSW~\cite{malkov_efficient_2018}). 
A challenge with indices like DiskANN is that the tradeoffs between accuracy, latency, and space are not well understood --- requiring users to empirically evaluate the indices under different hyperparameter settings, until the right tradeoffs are found.

To address the challenge, we propose $\rpt$, an efficient post-hoc pruning algorithm derived from the $\rprune$ subroutine used in the DiskANN index construction. $\rpt$ enables dynamic adjustment of DiskANN's reachability parameter, giving a tunable trade-off between accuracy, latency, and index size. By starting from a high-quality base graph, it avoids the cost of repeated index reconstruction. Leveraging recent theoretical frameworks for $\alpha$-reachable graphs~\cite{indyk2023worstcaseperformancepopularapproximate, gollapudi2025sort}, we provide a rigorous analysis of $\rpt$. Namely, we prove that a graph pruned via $\rpt$ preserves specific worst-case reachability guarantees in general metrics while achieving superior worst-case bounds in Euclidean metrics.

Our empirical evaluation across four public datasets demonstrates that $\rpt$ accelerates the DiskANN tuning process by up to 43$\times$. Moreover, indices tuned with $\rpt$ exhibit performance gains over those rebuilt with the same reachability parameters. Similarly to model distillation, our techniques allow for the ``distillation'' of an index into different configurations, suitable for different hardware capabilities, with predictable accuracy.

\section{Preliminaries}
We consider a dataset $P$ in a metric space $(X,D)$ where $X$ is the universe of points and $D$ is the distance function. We use $\lambda$ and $\Delta$ to denote the doubling dimension and the aspect ratio of $P$. A graph-based ANN index for $P$ is given by a directed graph $G = (P,E)$; we use $(a,b)$ to denote the one-way edge from $a$ to $b$. Note that we will abuse $P$ to refer to the dataset and the vertex set of the graph, since in this paper, points in the dataset have a one-to-one correspondence to vertices in the graph. We use $N_{out}(a), N_{in}(a)$ to denote the set of out-neighbors and in-neighbors, respectively, of point $a \in P$ when the graph in the context is clear. 

We give a brief review of the DiskANN algorithm and its theoretical guarantees below. We refer readers to \cite{NEURIPS2019_09853c7f, indyk2023worstcaseperformancepopularapproximate, gollapudi2025sort} for full details.

\subsection{DiskANN Algorithm Review}
There are multiple variants of DiskANN \cite{10.1145/3543507.3583552, jaiswal2022ooddiskannefficientscalablegraph, adams2025distributedannefficientscalingsingle}. We review two variants that are relevant to this paper: the original Vamana heuristics proposed in~\cite{NEURIPS2019_09853c7f} and the  ``slow preprocessing'' variant formalized in \cite{indyk2023worstcaseperformancepopularapproximate}. 

The search algorithm is the same for both variants (see \Cref{alg:greedySearch}). Given the index graph, the search is initialized by seeding a candidate set $\mathcal{L}$ with a designated start vertex $s$ (typically the medoid) and maintaining an empty set of visited vertices $\mathcal{V}$. In each iteration, the algorithm explores the graph by selecting the point $p^*$ from the unexplored candidates in $\mathcal{L}$ that is most proximal to the query $q$. By immediately adding the out-neighbors $N_{\text{out}}(p^*)$ to the candidate set $\mathcal{L}$ and marking $p^*$ as visited, the algorithm effectively ``walks'' through the graph's edges toward the query. To ensure the search remains computationally tractable, the candidate set $\mathcal{L}$ is restricted to a maximum size $L$; if the addition of neighbors causes the set to exceed this bound, $\mathcal{L}$ is pruned to retain only the $L$ points closest to $q$. This iterative process of greedy selection and neighborhood expansion continues until no unvisited points remain in the search beam, at which point the $k$ best-found candidates are returned.

The index construction of both variants relies on the subroutine $\rprune$ (see \Cref{alg:robustprune}). $\rprune$ prunes out-neighbor candidates for the input vertex $p$ and utilizes a parameter $\alpha \ge 1$ to favor ``diverse'' neighbors that provide long-range shortcuts across the dataset. The algorithm proceeds by iteratively selecting the point $p^*$ in the candidate set $\mathcal{V}$ that is closest to the target point $p$. Once a candidate is added to $N_{out}(p)$, a pruning condition is applied to all remaining candidates: any point $p'$ that is closer to the newly added out-neighbor $p^*$ than it is to the target point $p$ by a factor of $\alpha$ (i.e., $D(p^*,p') < D(p,p') / \alpha$) is removed from the candidate set $\mathcal{V}$. The intuition is that the edge $(p,p')$ is potentially not important if moving from p to p* decreases the distance to p' by a factor of $\alpha$. 

The ``slow preprocessing'' construction executes $\rprune(p, P \setminus \{p\}, \alpha, n - 1)$ for every vertex $p$ in the graph. Notice the candidate set $\mathcal{V} = P \setminus \{p\}$, and the degree limit $R = n - 1$, which results in an $O(n^3)$ time complexity that is impractical at scale. However, it provides a rigorous structural foundation, $\alpha$-reachability, for analyzing the algorithm's worst-case performance on construction time, query time, and approximation factor.  
\begin{definition}[\textbf{$\alpha$-reachable and $\alpha$-reachability}]
Let $\alpha > 1$. We say, in a graph $G = (P, E)$, $q$ is $\alpha$-reachable from another vertex $p$ if either $(p, q) \in E$, or there exists $p'$ s.t. $(p, p') \in E$ and $D(p', q) \cdot \alpha \leq D(p, q)$. We say a graph $G$ is $\alpha$-reachable if any vertex in $P$ is $\alpha$-reachable from any other vertex in $P$. 

Moreover, we define the reachability of a graph $G$ as the maximum number $\alpha^*$ such that $G$ is $\alpha^*$-reachable.
\end{definition}
Then, the following result on the ``slow preprocessing'' construction can be proved. Recall $\lambda$ and $\Delta$ denote the dataset's doubling dimension and the aspect ratio.
\begin{lemma}
[Lemma 3.2 and 3.3 in \cite{indyk2023worstcaseperformancepopularapproximate}]
\label{lem:degree}
    Applying $\rprune(p, P \setminus \{p\}, \alpha, \infty)$ for every $p \in P$ yields an $\alpha$-reachable graph with maximum degree $O((4\alpha)^\lambda \log \Delta)$.
\end{lemma}
They also provided a lower bound on worst-case search performance, which is 
later improved by \cite{gollapudi2025sort}. In particular, we have  

\begin{theorem}[Theorem 1.1 in \cite{gollapudi2025sort}]
\label{thm:sortbeforeprune}
Let $G$ be a $\alpha$-reachable graph. For any query $q$, any start point $s$, any constant $\epsilon > 0$, and any $k < L$, $\textsc{GreedySearch}(s, q, k, L)$ in at most $O\left(L + \log_\alpha \frac{\Delta}{(\alpha-1)\epsilon}\right)$ steps outputs a set of points $\{b_1, \dots, b_L\}$ such that each $b_j$ satisfies:
\begin{align*}
    D(b_j, q) &\le \epsilon + \frac{\alpha}{\alpha - 1} \cdot D(a_j, q) \quad \text{for } \ell_2 \text{ metric } D \\
    D(b_j, q) &\le \epsilon + \frac{\alpha + 1}{\alpha - 1} \cdot D(a_j, q) \quad \text{for any metric } D
\end{align*}
where $a_j$ is the $j$th nearest neighbor to $q$.
\end{theorem}

From the above, we can see the parameter $\alpha$ is a key parameter for the worst-case trade-off of the ``slow preprocessing'' DiskANN (an $\alpha$-reachable graph), which has
\begin{itemize}
    \item Storage: $O(n\cdot (4\alpha)^\lambda \log \Delta)$
    \item Query Time: $O\left(\left(L + \log_\alpha \frac{\Delta}{(\alpha-1)\epsilon}\right)\cdot(4\alpha)^\lambda \log \Delta\right)$
    \item Approximation: $\left(\epsilon + \frac{\alpha}{\alpha - 1}\right)$ for Euclidean metrics; $\left(\epsilon + \frac{\alpha + 1}{\alpha - 1}\right)$ for general metrics.
\end{itemize}
As $\alpha$ increases, storage and query time increase while the approximation factor decreases, indicating more accurate output. 

The original Vamana index construction, on the other hand, caps the maximum out-degree of the graph for fast construction in practice. Besides $\rprune$, it uses other heuristics, which we will not discuss since we never open the black box in this paper and we refer readers to \cite{NEURIPS2019_09853c7f} for details. Despite no clear worst-case guarantee, Vamana can achieve state-of-the-art empirical performance, and the parameter $\alpha$ influences the empirical trade-off in the same directions as it does in the theoretical framework of the slow preprocessing variant. 

\begin{algorithm}[h]
   \caption{\rprune($p, \mathcal{V}, \alpha, R$)}
   \label{alg:robustprune}
\begin{algorithmic}
   \STATE {\bfseries Input:} Graph $G$, point $p \in P$, candidate set $\mathcal{V}$, desired reachability $\alpha \ge 1$, degree bound $R$
   \STATE {\bfseries Output:} $G$ is modified by setting at most $R$ new out-neighbors for $p$
   \STATE $\mathcal{V} \leftarrow (\mathcal{V} \cup N_{\text{out}}(p)) \setminus \{p\}$
   \STATE $N_{\text{out}}(p) \leftarrow \emptyset$
   \WHILE{$\mathcal{V} \neq \emptyset$}
   \STATE $p^* \leftarrow \arg\min_{p' \in \mathcal{V}} D(p, p')$
   \STATE $N_{\text{out}}(p) \leftarrow N_{\text{out}}(p) \cup \{p^*\}$
   \IF{$|N_{\text{out}}(p)| = R$}
   \STATE {\bfseries break}
   \ENDIF
   \FOR{$p' \in \mathcal{V}$}
   \IF{$\alpha \cdot D(p^*, p') \le D(p, p')$}
   \STATE Remove $p'$ from $\mathcal{V}$
   \ENDIF
   \ENDFOR
   \ENDWHILE
\end{algorithmic}
\end{algorithm}

\begin{algorithm}[h]
   \caption{GreedySearch($s, q, k, L$)}
   \label{alg:greedySearch}
\begin{algorithmic}
   \STATE {\bfseries Input:} Graph $G$ with start vertex $s$, query $q$, result size $k$, search list size $L \ge k$
   \STATE {\bfseries Output:} Result set $\mathcal{L}$ containing $k$-approx NNs, and a set $\mathcal{V}$ containing all the visited vertices
   \STATE Initialize sets $\mathcal{L} \leftarrow \{s\}$ and $\mathcal{V} \leftarrow \emptyset$
   \WHILE{$\mathcal{L} \setminus \mathcal{V} \neq \emptyset$}
   \STATE Let $p^* \leftarrow \arg\min_{p \in \mathcal{L} \setminus \mathcal{V}} D(p, q)$
   \STATE Update $\mathcal{L} \leftarrow \mathcal{L} \cup N_{\text{out}}(p^*)$ and $\mathcal{V} \leftarrow \mathcal{V} \cup \{p^*\}$
   \IF{$|\mathcal{L}| > L$}
   \STATE Update $\mathcal{L}$ to retain closest $L$ points to $q$
   \ENDIF
   \ENDWHILE
   \STATE {\bfseries return} [closest $k$ points from $\mathcal{L}$; $\mathcal{V}$]
\end{algorithmic}
\end{algorithm}

\section{$\rpt$ and Worst-Case Guarantees}
In this section, we propose $\rpt$, a simple yet disciplined pruning algorithm based on $\rprune$, which efficiently prunes an $\alpha$-reachable graph $G$ into a sparser graph with a theoretical guarantee of worst-case performance trade-off. 

\paragraph{\rpt} $\rpt$ takes an initially $\alpha_1$-reachable graph $G = (P, E)$, and a new reachability target $\alpha_2 < \alpha_1$ as inputs. It simply runs $\rprune(p, N_{out}(p), \alpha_2, \infty)$ for every $p \in P$ and outputs the modified graph $G' = (P, E')$. 

We study the worst-case reachability of the output $G'$ of $\rpt(G, \alpha_2)$ where $G$ is an $\alpha_1$-reachable graph. It turns out that the worst-case reachability of $G'$ is a function of $\alpha_1, \alpha_2$ and depends on the metric space and whether the pruning subroutine $\rprune$ sorts the out-neighbors. To distinguish the unsorted version from $\rprune$, we use $\urprune$ to denote the procedure that is the same as $\rprune$ but only replaces the step ``$p^* \leftarrow \arg\min_{p' \in \mathcal{V}} D(p, p')$" in $\Cref{alg:robustprune}$ with ``$p^*$ being an arbitrary out-neighbor of $p$". Our results on the reachability of $G'$ in different settings are summarized in \Cref{tb:new-alpha}, and formally stated in \Cref{thm:unsorted} and \Cref{thm:sorted}. 

\begin{table}[h]
\caption{Worst-case reachability after $\rpt$}
\label{tb:new-alpha}
  \begin{center}
    \begin{small}
      \begin{sc}
        \renewcommand{\arraystretch}{2}
        \begin{tabular}{l|cc}
          \toprule
           & General & Euclidean \\
          \midrule
          Unsorted    & $\frac{1}{\frac{1}{\alpha_1} + \frac{1}{\alpha_2} + 1}$   & $\frac{1}{\frac{1}{\alpha_1} + \frac{1}{\alpha_2} + 1}$   \\
          Sorted    & $\frac{1}{\frac{1}{\alpha_1} + \frac{1}{\alpha_2} }$   & $\frac{1}{\frac{1}{\alpha_1}\cdot \sqrt{1 - \frac{1}{4\alpha_2^2}} + \frac{1}{\alpha_2}\cdot \sqrt{1 - \frac{1}{4\alpha_1^2}}}$ \\
          \bottomrule
        \end{tabular}
      \end{sc}
    \end{small}
  \end{center}
  \vskip -0.1in
\end{table}
 Note $1/({\frac{1}{\alpha_1}\cdot \sqrt{1 - \frac{1}{4\alpha_2^2}} + \frac{1}{\alpha_2}\cdot \sqrt{1 - \frac{1}{4\alpha_1^2}}})$ is larger than $\frac{1}{\frac{1}{\alpha_1} + \frac{1}{\alpha_2} }$. For a concrete example, when $\alpha_1 = 3, \alpha_2 = 2$, this expression $\approx 1.226$ while $\frac{\alpha_1 \alpha_2}{\alpha_1 + \alpha_2} = 1.2$.

\begin{theorem}[Unsorted Reachability]
\label{thm:unsorted}
Let $P$ be a dataset in any metric space $(X, D)$. Fix finite parameters $\alpha_1 \ge \alpha_2 > 1$. Let $G = (P, E)$ be any $\alpha_1$-reachable graph on $P$. Let $G' = (P, E')$ be the graph obtained by applying $\textsc{RobustPrune}(p, N_{\text{out}}(p), \alpha_2, \infty)$ for all $p \in P$. Then, the graph $G'$ has reachability $ 1/\left(1 / \alpha_1 + 1 / \alpha_2 + 1\right)$ in the worst case.
\end{theorem}
\newcommand{\betagen}{\beta_{\textup{gen}}}
\newcommand{\betaeuc}{\beta_{\textup{euc}}}

\begin{theorem}[Sorted Reachability]
\label{thm:sorted}
Let $P$ be a dataset in a metric space $(X,D)$. Fix finite parameters $\alpha_1 \ge \alpha_2 > 1$. Let $G = (P, E)$ be any $\alpha_1$-reachable graph on $P$. Let $G' = (P, E')$ be the graph constructed by applying $\textsc{RobustPrune}(p, N_{\text{out}}(p), \alpha_2, \infty)$ for all $p \in P$. 
Then the following holds in the worst case:
\begin{itemize} 
    \item For any metric space, $G'$ has reachability $1/\left(1 / \alpha_1 + 1 / \alpha_2 \right)$. 
    \item If $(X,D)$ is a Euclidean space, then $G'$ has reachability $1/ \left({(1/\alpha_1) \cdot \sqrt{1-\frac{1}{4\alpha_2^2}} + (1/\alpha_2) \cdot \sqrt{1-\frac{1}{4\alpha_1^2}}}\right)$.
\end{itemize}
\end{theorem}
Readers may find an analogy in our results to \cite{gollapudi2025sort}, where Euclidean metrics and neighbor sorting guarantee a better worst-case reachability. In our problem, neighbor sorting alone already admits a better reachability, and it is further improved in Euclidean metrics. However, Euclidean metrics alone do not admit better reachability.

The reason why the reachability of the pruned graph $G'$ is crucial is as follows. First, the following lemma in \cite{indyk2023worstcaseperformancepopularapproximate} can help us to bound the maximum degree of the pruned graph $G'$. 

\begin{lemma}[Lemma 3.3 in \cite{indyk2023worstcaseperformancepopularapproximate}]
    For any $\alpha > 1$, graph $G = (P, E)$, and $p \in P$, applying $\rprune(p, N_{out}(p), \alpha, \infty)$ to every $p \in P$ guarantees that the maximum degree of the modified graph is $O((4\alpha)^\lambda \log \Delta)$.
\end{lemma}

Using the above lemma, we know that the size of $G'$ is $n \cdot O((4\alpha_2)^\lambda \log \Delta)$. The query time of $G'$ now only depends on the number of steps during search, and the approximation of $G'$ can be quantified by its reachability via \Cref{thm:sortbeforeprune}, where the larger the reachability is, the smaller both the number of steps and approximation are. Therefore, the rest of analysis on query time and approximation is reduced to exploring the reachability of $G'$. We want it to be as large (or as close to $\alpha_2$) as possible. 

Next, we give the proofs of the above two theorems and include the skipped proof in the appendix. 

\subsection{Proof of \Cref{thm:unsorted}}
Let $G = (P,E)$ be an $\alpha_1$-reachable graph and let $G' = (P,E')$ be obtained by applying $\urprune(p, \emptyset, \alpha_2, \infty)$ at every $p \in P$.

Let $\alpha^*$ denote the reachability of $G'$. Fix an ordered pair $p,z$ that is tight for $\alpha^*$-reachability and satisfies $(p,z)\notin E'$ (such a pair must exist since $\alpha^*$ is finite). Then, the directed edge $(p,z)$ was either (a) absent in $G$, or (b) pruned from $G$. In case (b), there is some edge $(p,y)\in E$ for which $D(y,z)\le D(p,z)/\alpha_2$, which gives $\alpha^\ast\ge \alpha_2$.

In case (a), $\alpha_1$-reachability of $G$ implies there is some $(p,y)\in E$ such that $D(y,z) \le D(p,z) / \alpha_1$. This means that applying $\urprune$ with parameter $\alpha_2$ either (a1) prunes $(p,y)$ from $G'$, or (a2) there is some edge $(p,x)\in E$ such that $D(x,y) \le D(p,y)/\alpha_2$. As before, in case (a1) it follows that $\alpha^*\ge \alpha_1$. In case (a2), the edge $(p,x)$ gives a lower bound on the reachability of $G'$. Namely, $\alpha^*$ is at least the value of the following optimization problem:

    \begin{align*}
    & \underset{p, x, y, z\in X}{\min}
    & & {D(p, z)} / {D(x, z)} \\
    & \text{subject to}
    & & D(y, z) \leq {D(p, z)} / {\alpha_1} \\
    & & & D(x, y) \leq {D(p, y)} / {\alpha_2} \\
    \end{align*}
By invariance under scaling, we may take $D(p,z) = 1$ and flip the objective to get the following crucial lemma. 

\newcommand{\opt}{\mathrm{OPT}}
\begin{restatable}{lemma}{reduceunsorted}
    \label[lemma]{lem:reduceunsorted}
    Let $(X, D)$ be a metric space. Given parameters $\alpha_1 \ge \alpha_2 \ge 1$, the worst-case reachability of the graph $G'$ obtained by taking an $\alpha_1$-reachable graph $G$ and applying $\urprune(p, \emptyset, \alpha_2, \infty)$ at every $p \in P$ is $\min\{\alpha_2, 1/\opt\}$  where $\opt$ is the value of the following optimization problem:
\begin{equation*}
\begin{aligned}
& \underset{p, x, y, z\in X}{\max}
& & D(x, z) \\
& \text{subject to}
& & D(y, z) \leq {1} / {\alpha_1} \\
& & & D(x, y) \leq {D(p, y)} / {\alpha_2}  \\
& & & D(p,z) = 1
\end{aligned}
\end{equation*}
\end{restatable}

We solve the optimization problem in a straightforward geometric way. We first apply the triangle inequality and the constraints in the above problem to get an upper bound of $D(x,z)$: 
\begin{align*}
    D(x,z) & \leq D(x,y) + D(y,z) \\
    & \leq D(p,y)/\alpha_2 + 1/\alpha_1 \\
    & \leq (D(p,z) + D(y,z))/\alpha_2 + 1/\alpha_1 \\
    & \leq (1 + 1/\alpha_1) / \alpha_2 + 1/\alpha_1 \\
    & = (\alpha_1 + \alpha_2 + 1)/(\alpha_1 \alpha_2),
\end{align*} 
where the first and third lines use the triangle inequality and the second and fourth lines use the constraints of the optimization problem. We also know $D(x,z)$ in any setting of $p,x,y,z$ that satisfies the problem constraints is a lower bound of the optimal objective since it is a maximization problem. It remains to show a setting where $ D(x,z) = (\alpha_1 + \alpha_2 + 1)/(\alpha_1 \alpha_2)$. 
\begin{lemma}
\label[lemma]{lem:unsorted-example}
    In a general metric space or a Euclidean space $(X,D)$, there exists a setting of four points $p,x,y,z \in X$ that satisfies $D(y,z) \leq 1/\alpha_1, D(x,y) \leq D(p,y)/\alpha_2, D(p,z) = 1$, and $D(x,z) = (\alpha_1 + \alpha_2 + 1)/(\alpha_1 \alpha_2)$.
\end{lemma}

\begin{proof}
    In the setting, $p,x,y,z$ are collinear and $D(p,z) = 1, D(z,y) = 1/\alpha_1, D(x,y) = (1 + 1/\alpha_1)/\alpha_2$. The constraints can be easily verified. We then compute $D(x,z) = 1/\alpha_1 + (1 + 1/\alpha_1)/\alpha_2 = (\alpha_1 + \alpha_2 + 1) / (\alpha_1 \alpha_2)$ as desired.
\end{proof}
The setting in the proof is shown in \Cref{fig:unsorted_example}.
\begin{figure}[ht]
  \begin{center}
    \begin{tikzpicture}[scale=0.73]
  \coordinate (p) at (-4,0);
  \coordinate (z) at (0,0);
  \coordinate (y) at (2.5,0);
  \coordinate (x) at (6,0);

  \draw[axis] (-4.6,0) -- (6.6,0);

  \node[pt, fill=black] (P) at (p) {};
  \node[pt, fill=blue!75!black] (Z) at (z) {};
  \node[pt, fill=green!65!black] (Y) at (y) {};
  \node[pt, fill=red!75!black] (X) at (x) {};

  \node[above=3.5pt] at (P) {$p$};
  \node[above=3.5pt, text=blue!75!black]  at (Z) {$z$};
  \node[above=3.5pt, text=green!65!black] at (Y) {$y$};
  \node[above=3.5pt, text=red!75!black]   at (X) {$x$};

  \draw[dim] ($(P)+(0,-1.05)$) -- node[dimlab, below=3.9pt] {\scriptsize$D(p,z)=1$} ($(Z)+(0,-1.05)$);
  \draw[dim] ($(Z)+(0,-1.05)$) -- node[dimlab, below=3pt] {\scriptsize$D(z,y)=\frac{1}{\alpha_1}$} ($(Y)+(0,-1.05)$);
  \draw[dim] ($(Y)+(0,-1.05)$) -- node[dimlab, below=3pt] {\scriptsize$D(y,x)=\frac{D(p,y)}{\alpha_2}$} ($(X)+(0,-1.05)$);
\end{tikzpicture}

    \caption{The setting that achieves the optimal objective in the optimization problem in \Cref{lem:reduceunsorted}.}
    \label{fig:unsorted_example}
  \end{center}
\end{figure}
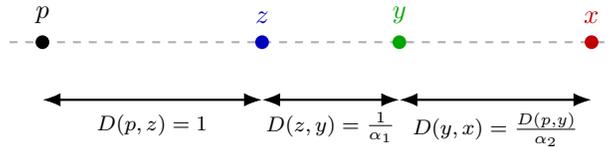
The setting exactly makes the inequalities used in constructing the upper bound tight. In conclusion, we show the optimization problem solves to $(\alpha_1 + \alpha_2 + 1) / (\alpha_1 \alpha_2)$. By \Cref{lem:reduceunsorted}, the worst-case reachability is $\min(\alpha_1, \alpha_2, (\alpha_1 \alpha_2) / (\alpha_1 + \alpha_2 + 1)) = (\alpha_1 \alpha_2) / (\alpha_1 + \alpha_2 + 1)$. That completes the proof for \Cref{thm:unsorted}. 

Also notice that, in \Cref{fig:unsorted_example}, $x$ is the furthest point from $p$, while $z$ is the closest one. Later we will see that this is a forbidden pattern when using $\rprune$ instead of $\urprune$, because $x, y , z$ have to be considered in order of their distances to $p$ in order to achieve an analogous case to $(a2)$.

\subsection{Proof Overview of \Cref{thm:sorted}}
Now we use \rprune where neighbors in the graph are considered for pruning in ascending order of their distances to the pivot point. The proof follows the same template as the unsorted case. The key difference is that, because $\rprune$ processes candidates in increasing distance from $p$,
any point $x$ that can prune $y$ must be considered no later than $y$, hence must satisfy $D(p,x)\le D(p,y)$.
Moreover, since $y$ serves as the $\alpha_1$-reachability witness for $(p,z)$, we also have $D(p,y)\le D(p,z)$. This yields the following lemma.

\begin{restatable}{lemma}{reducesorted}
    \label[lemma]{lem:reducesorted}
    Let $(X, D(\cdot, \cdot))$ be a metric space, and $p,x,y,z \in X$ be four points. Given parameters satisfying $\alpha_1 \ge \alpha_2 \ge 1$, the worst-case reachability of the graph $G'$ obtained by taking an $\alpha_1$-reachable graph and applying $\rprune(p, \emptyset, \alpha_2, \infty)$ at every $p \in P$ is $\min\{\alpha_2,1/\opt\}$, where $\opt$ is the value of the following optimization problem:
\begin{equation*}
\begin{aligned}
& \underset{p, x, y, z\in X}{\max}
& & D(x, z) \\
& \text{subject to}
& & D(y, z) \leq {1} / {\alpha_1} \\
& & & D(x, y) \leq {D(p, y)} / {\alpha_2}  \\
& & & D(x, y) \leq D(p,y) \leq D(p,z) = 1 
\end{aligned}
\end{equation*}
\end{restatable}

Notice that the only difference in the optimization problem is that there is an ordered condition in the last constraint, which rules out the pattern seen in the previous subsection. However, we can obtain an upper bound on the optimal objective of the problem in the same way using the triangle inequality and problem constraints: $D(x,z) \leq D(x,y) + D(y,z) \leq D(p,y)/\alpha_2 + 1/\alpha_1 \leq 1/\alpha_2 + 1/\alpha_1 = (\alpha_1 + \alpha_2) / (\alpha_1 \alpha_2)$. Note that the upper bound is smaller than the one in unsorted case. In the general metric space, the upper bound is actually tight.
\begin{lemma}
\label[lemma]{lem:sorted-general-example}
    In a general metric space $(X,D)$, there exists a setting of four points $p,x,y,z \in X$ that satisfies $D(y,z) \leq 1/\alpha_1, D(x,y) \leq D(p,y)/\alpha_2, D(x, y) \leq D(p,y) \leq D(p,z) = 1$, and $D(x,z) = (\alpha_1 + \alpha_2)/(\alpha_1 \alpha_2)$.
\end{lemma}

\begin{proof}
    The setting is defined by the following distance matrix: 
\[
\renewcommand{\arraystretch}{1.5} 
\setlength{\tabcolsep}{10pt}      
\begin{array}{r|cccc}
   & p & x & y & z \\
\hline
 p & 0 &   &   &   \\
 x & 1 & 0 &   &   \\
 y & 1 & \frac{1}{\alpha_1} & 0 &   \\
 z & 1 & \frac{\alpha_1 + \alpha_2}{\alpha_1 \alpha_2} & \frac{1}{\alpha_2} & 0 \\
\end{array}
\]

It is easy to verify that the triangular inequalities and the problem constraints are satisfied.
\end{proof}

\Cref{fig:sorted}(a) shows the setting in the proof, where $x,y,z$ are on the (general metric) unit sphere centered at $p$. Since the lower bound of $D(x,z)$ given by the above setting matches the upper bound, we can conclude that the worst-case reachability is $(\alpha_1 \alpha_2) / (\alpha_1 + \alpha_2)$ for general metric spaces by \Cref{lem:reducesorted}. 

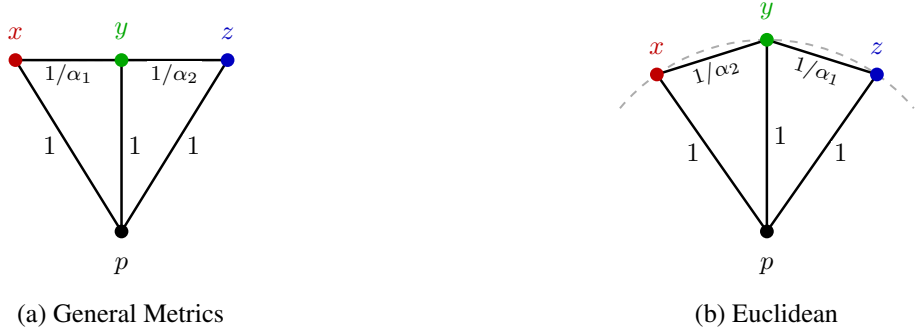
\begin{figure}[ht]
  \begin{center}
  \centering
  \begin{subfigure}[t]{0.48\linewidth}
    \centering
    \begin{tikzpicture}[scale=0.6]
      \coordinate (P) at (0,-2.25);
      \coordinate (Y) at (0, 1.55);
      \coordinate (X) at (-2.35, 1.55);
      \coordinate (Z) at ( 2.35, 1.55);

      \draw[axis] (X) -- (Y) -- (Z);

      \draw[dist] (P) -- (X) node[lab, midway, xshift=-0.25cm] {$1$};
      \draw[dist] (P) -- (Y) node[lab, midway, xshift=0.18cm] {$1$};
      \draw[dist] (P) -- (Z) node[lab, midway, xshift=0.25cm] {$1$};

      \draw[dist] (X) -- (Y) node[midway, yshift=-0.2cm] {\scriptsize$1/\alpha_1$};
      \draw[dist] (Y) -- (Z) node[lab, midway, yshift=-0.2cm] {\scriptsize$1/\alpha_2$};

      \node[pt, fill=black] (p) at (P) {};
      \node[pt, fill=red!75!black] (x) at (X) {};
      \node[pt, fill=green!65!black] (y) at (Y) {};
      \node[pt, fill=blue!75!black] (z) at (Z) {};

      \node[below=6pt] at (p) {$p$};
      \node[above=4pt, text=red!75!black]   at (x) {$x$};
      \node[above=4pt, text=green!65!black] at (y) {$y$};
      \node[above=4pt, text=blue!75!black]  at (z) {$z$};
    \end{tikzpicture}

    \caption{General Metrics}
  \end{subfigure}
  \hfill
  \begin{subfigure}[t]{0.48\linewidth}
    \centering
    \begin{tikzpicture}[scale=0.6]
      \coordinate (P) at (0,-2.25);
    \def\r{4.25} 

    \path (P) ++(125:\r) coordinate (X);
    \path (P) ++( 90:\r) coordinate (Y);
    \path (P) ++( 55:\r) coordinate (Z);

    \draw[axis] (P) ++(140:\r) arc[start angle=140, end angle=40, radius=\r];

      \draw[dist] (P) -- (X) node[midway, xshift=-0.25cm] {$1$};
      \draw[dist] (P) -- (Y) node[midway, xshift=0.18cm] {$1$};
      \draw[dist] (P) -- (Z) node[lab, midway, xshift=0.25cm] {$1$};

      \draw[dist] (X) -- (Y) node[midway,  yshift=-0.2cm,sloped] {\scriptsize$1/\alpha_2$};
      \draw[dist] (Y) -- (Z) node[midway, yshift=-0.2cm, sloped] {\scriptsize$1/\alpha_1$};

      \node[pt, fill=black] (p) at (P) {};
      \node[pt, fill=red!75!black] (x) at (X) {};
      \node[pt, fill=green!65!black] (y) at (Y) {};
      \node[pt, fill=blue!75!black] (z) at (Z) {};

      \node[below=6pt] at (p) {$p$};
      \node[above=4pt, text=red!75!black]   at (x) {$x$};
      \node[above=4pt, text=green!65!black] at (y) {$y$};
      \node[above=4pt, text=blue!75!black]  at (z) {$z$};
    \end{tikzpicture}

    \caption{Euclidean}
  \end{subfigure}
    \caption{The settings that achieve optimal objectives in the optimization problem in \Cref{lem:reducesorted} for (a) general metric spaces and (b) Euclidean metrics.}
    \label{fig:sorted}
  \end{center}
\end{figure}

Unfortunately, the optimal setting for general metric spaces that requires $x,y,z$ to be simultaneously collinear and on a unit sphere is impossible in Euclidean metrics. This indicates that the upper bound derived only using the triangle inequality and the problem constraints is likely to be loose in Euclidean space. We resort to tools for semi-definite programming to solve for Euclidean space as in \cite{gollapudi2025sort} and get the following lemma.

\begin{restatable}{lemma}{sortedeuclidopt}
\label[lemma]{lem:sorted-euclid-opt}
Let $\alpha_1\ge \alpha_2> 1$. The value of the optimization problem
\begin{equation*}
\begin{aligned}
\max_{x,y,z\in\mathbb{R}^d}\quad & \|x-z\|_2 \\
\text{s.t.}\quad
& \|y-z\|_2 \le \frac{\|z\|_2}{\alpha_1},\\
& \|x-y\|_2 \le \frac{\|y\|_2}{\alpha_2},\\
& \|x\|_2 \le \|y\|_2 \le \|z\|_2 = 1
\end{aligned}
\end{equation*}
is
\[
\frac{1}{\alpha_1}\sqrt{1-\frac{1}{4\alpha_2^2}}
\;+\;
\frac{1}{\alpha_2}\sqrt{1-\frac{1}{4\alpha_1^2}}.
\]
\end{restatable}

can still be obtained by $p,x,y,z$ on the same 2-dimensional plane in Euclidean metrics. In the optimal setting (\Cref{fig:sorted}(b)), intuitively speaking, the curvature of Euclidean unit sphere pushes $x, z$ towards each other, which results in a smaller maximum $D(x,z)$, thus a larger worst-case reachability for sorted Euclidean metrics.  

\section{Experimental Evaluation}

We have shown theoretical guarantees of the $\alpha$-reachable graphs pruned by $\rpt$. In this section, we test $\rpt$ on DiskANN indices to evaluate its real-world performance. 
We will use $100$-recall@$100$, a widely adopted metric in practice, to measure query accuracy. \recall is the proportion of the true top $100$ nearest neighbors contained in the top $100$ results returned by the algorithm. 

\subsection{Setup}

\textbf{Hardware.} We conduct all experiments on a cloud instance equipped with an AMD EPYC 9B14 CPU (15 cores, 30 threads) and 256GB RAM, running Debian GNU/Linux 12.

\textbf{Datasets.} We evaluate our methods on four benchmark datasets obtained from the BigANN benchmark repository~\citep{simhadri2022results}.
\textbf{SIFT-1M} and \textbf{GIST-1M}~\citep{jegou2011product} are classical computer vision datasets; the former consists of 1 million 128-dimensional SIFT feature vectors with 10,000 query vectors, while the latter contains 1 million 960-dimensional GIST descriptors capturing global image structure, with 1,000 query vectors.
\textbf{Deep-1M}~\citep{babenko2016efficient} comprises 1 million 96-dimensional vectors derived from deep neural network embeddings, with 10,000 query vectors.
Finally, \textbf{MSSPACEV-1M} is a subset of 1 million 100-dimensional vectors from Microsoft's Bing query embeddings~\citep{simhadri2022results}, originally stored as int8 and converted to float32, with 29,316 query vectors.
All datasets use Euclidean distance as the similarity metric. 
\paragraph{Tuning Scheme Configuration}
We implemented the original DiskANN (Vamana)\footnote{The index is called Vamana in \cite{NEURIPS2019_09853c7f} and ``fast preprocessing version'' in \cite{indyk2023worstcaseperformancepopularapproximate}.} in \cite{NEURIPS2019_09853c7f} in Python. We construct the DiskANN indices for the above four datasets with the same parameter setting in \cite{NEURIPS2019_09853c7f}: $\alpha_1 = 1.2, R = 70, L-Build = 75$. They are referred to as the base or meta indices in the rest of this section. Then we tune the base indices via $\rpt$ and rebuild with $\alpha_2 = 1.1, 1.05, 1.01$. During query time, we use $\Cref{alg:greedySearch}$ with the start vertex $s$ being the medoid of the dataset, $k = 100$ (since we will use \recall to measure accuracy). We vary the index search effort by tuning the beam size parameter $L$ of \Cref{alg:greedySearch} as a knob to obtain different trade-offs between recall and query latency. 

\subsection{Comparison between $\rpt$ and Rebuild for Tuning}
\begin{table}[ht]
\centering
\caption{Runtime comparison: Rebuild vs.\ \rpt (seconds). The rightmost columns show the cumulative time and overall speedup (Total $T_{\text{rebuild}} / \text{Total } T_{\text{prune}}$) across all three configurations.}
\label{tab:rebuild-vs-prune-sec}
\resizebox{\textwidth}{!}{
\begin{tabular}{lrrrrrrrrrrrr}
\toprule
& \multicolumn{3}{c}{$\alpha=1.01$} & \multicolumn{3}{c}{$\alpha=1.05$} & \multicolumn{3}{c}{$\alpha=1.10$} & \multicolumn{3}{c}{\textbf{Total (All Three $\boldsymbol{\alpha}$)}} \\
\cmidrule(lr){2-4} \cmidrule(lr){5-7} \cmidrule(lr){8-10} \cmidrule(lr){11-13}
\textbf{Dataset} & $T_{\text{rebuild}}$ & $T_{\text{prune}}$ & Speedup & $T_{\text{rebuild}}$ & $T_{\text{prune}}$ & Speedup & $T_{\text{rebuild}}$ & $T_{\text{prune}}$ & Speedup & $T_{\text{rebuild}}$ & $T_{\text{prune}}$ & Speedup \\
\midrule
SIFT1M & 5,728 & 409 & 14$\times$ & 6,220 & 421 & 15$\times$ & 6,728 & 482 & 14$\times$ & 18,676 & 1,312 & 14$\times$ \\
DEEP1M & 5,536 & 376 & 15$\times$ & 5,911 & 499 & 12$\times$ & 6,745 & 555 & 12$\times$ & 18,192 & 1,430 & 13$\times$ \\
GIST1M & 12,073 & 294 & 41$\times$ & 14,665 & 367 & 40$\times$ & 21,312 & 462 & 46$\times$ & 48,050 & 1,123 & 43$\times$ \\
MSSPACEV1M & 7,573 & 492 & 15$\times$ & 9,028 & 560 & 16$\times$ & 11,570 & 495 & 23$\times$ & 28,171 & 1,547 & 18$\times$ \\
\bottomrule
\end{tabular}
}
\end{table}

\begin{figure}[ht]
  \begin{center}
\centerline{\includegraphics[width=\textwidth]{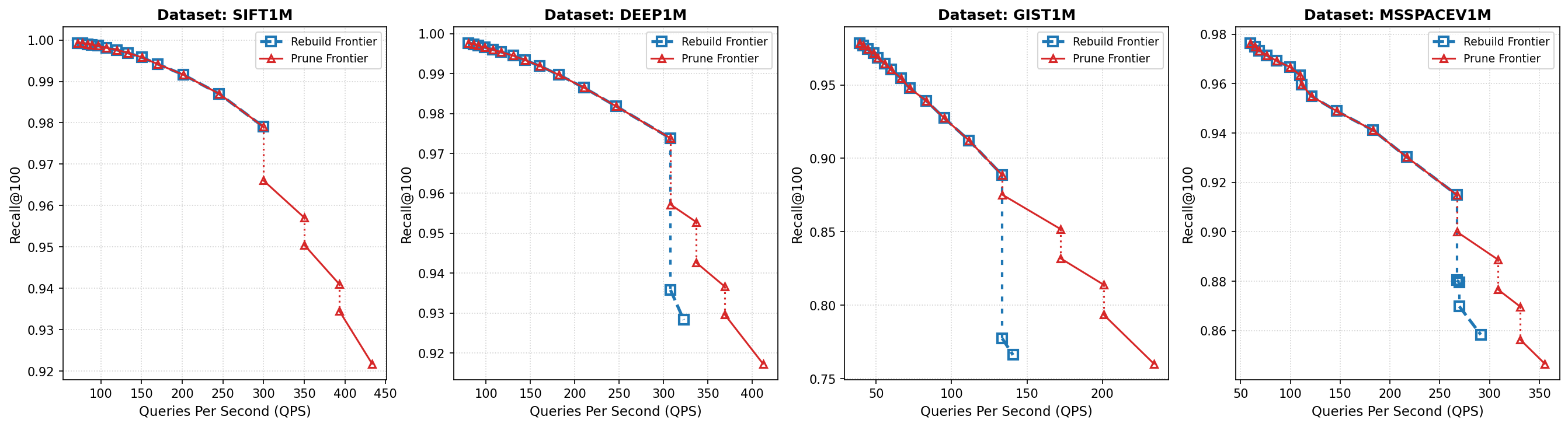}}
    \caption{Recall-QPS trade-off frontiers achieved by a base DiskANN graph of $\alpha_1 = 1.2$ then pruned graphs from the base graph via $\rpt$ with $\alpha_2 = 1.1, 1.05, 1.01$, and the same base DiskANN graph of $\alpha = 1.2$ along with rebuilt DiskANN graphs of $\alpha = 1.1, 1.05, 1.01$. }
    \label{fig:tuning-compare}
  \end{center}
\end{figure}

We compare the performance of $\rpt$ and rebuild for tuning the base index with three new alpha values by total processing time cost and recall-QPS trade-off.

Table~\ref{tab:rebuild-vs-prune-sec} demonstrates the tuning efficiency of $\rpt$. Our method achieves significant speedups compared to the rebuilding baseline across all datasets. Notably, the speed-up is the most significant on GIST1M (achieving up to $43\times$ total speedup). We attribute this performance to GIST1M’s inherently sparser graph structure: at $\alpha=1.2$, the index averages only 29 edges per node, compared to 52 for MsSpace1M. While this lower edge density significantly reduces the pruning workload for $\rpt$, the rebuilding baseline fails to exploit this structural sparsity because it remains latent to the algorithm until the graph construction is complete, resulting in substantially higher computational overhead.

\Cref{fig:tuning-compare} illustrates the Pareto frontiers for the Recall-QPS trade-off, derived by combining the base index with either pruned indices or rebuilt indices. The composite frontier traces the base index curve to its maximum achievable QPS before transitioning to the curve segment of the next pruning or rebuilding configuration. Across all four datasets, the frontier generated by $\rpt$ (Prune Frontier) consistently outperforms the Rebuild Frontier. Notably, on the SIFT1M dataset, the rebuilt indices are entirely dominated by the base graph and fail to extend the efficiency frontier. These prune frontiers can be reconstructed by connecting the upper-right segments of the individual curves presented in \Cref{fig:tuning}.

\subsection{Performance of Individual Indices Returned by $\rpt$}
\begin{figure}[ht]
  \begin{center}
\centerline{\includegraphics[width=\textwidth]{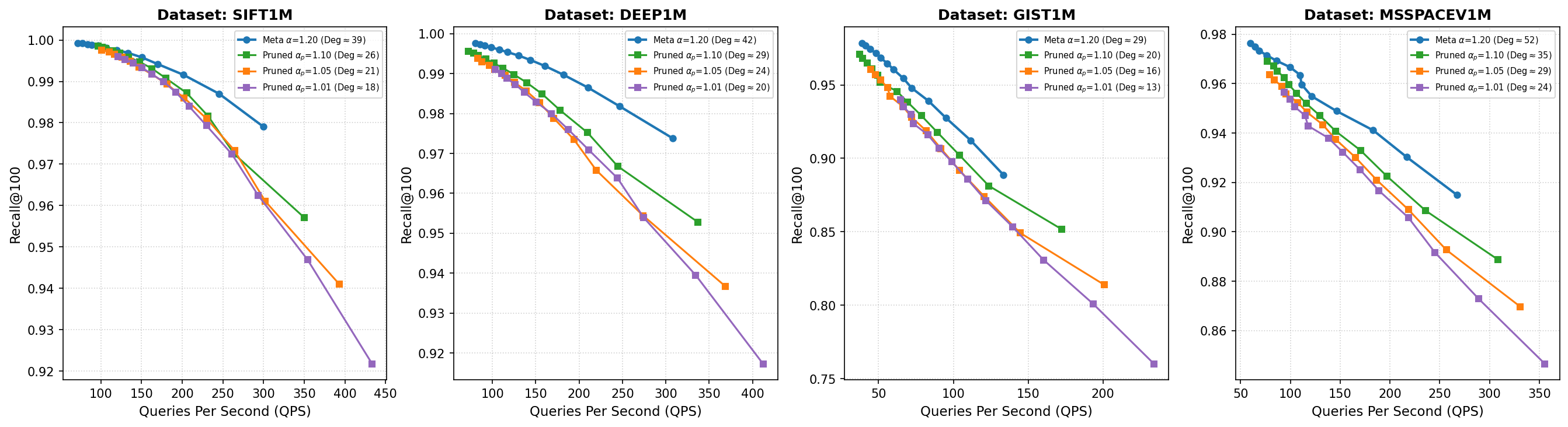}}
    \caption{Recall-QPS trade-off frontiers achieved by a base DiskANN graph of $\alpha_1 = 1.2$ (blue curves with circles) and pruned graphs (curves with squares) from the base graph of $\alpha_2 = 1.1, 1.05, 1.01$. Average degrees (Deg) of individual graphs are also included.}
    \label{fig:tuning}
  \end{center}
\end{figure}

\begin{figure}[!t]
  \begin{center}    \centerline{\includegraphics[width=\textwidth]{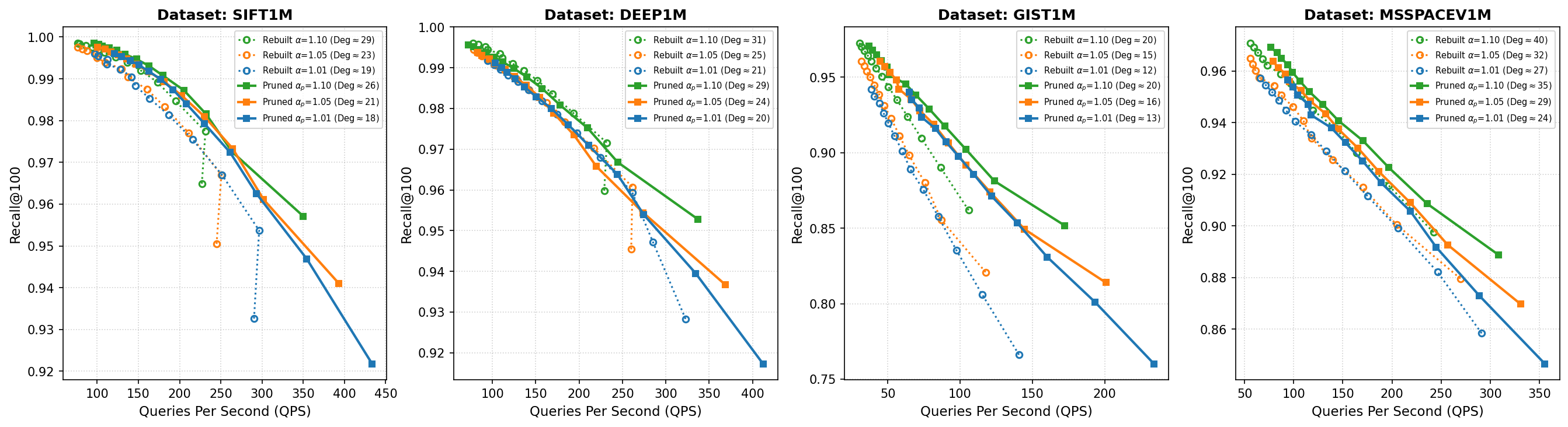}}
    \caption{Compare Recall-QPS performance of pruned (solid curves) and rebuilt (dashed curves) DiskANN indices. Pruned and rebuilt indices of the same $\alpha$ value share a color. Average degrees (Deg) of individual graphs are also included.}
    \label{fig:comparing}
  \end{center}
\end{figure}

We also investigate the performance of individual indices produced by $\rpt$ and rebuild during the tuning experiments, in order to obtain a better understanding of the effects of $\rpt$. 

Figure~\ref{fig:tuning} illustrates the performance of indices returned by $\rpt$ compared to the base indices. Starting with a base DiskANN index constructed with $\alpha=1.2$, we apply progressive pruning using $\alpha \in \{1.1, 1.05, 1.01\}$. The results demonstrate a clear trade-off between retrieval accuracy and query throughput: reducing $\alpha$ increases graph sparsity, which yields higher Queries Per Second (QPS) at the cost of marginal drops in Recall@100. While the denser graph excels in high-recall regimes, the sparser variants significantly reduce memory overhead and extend the maximum QPS, making them highly effective for resource-constrained local deployments or latency-critical applications. The consistent deviation of the pruned index curves from the baseline frontier suggests that $\rpt$ preserves the structural integrity of the base graph during pruning.

Figure~\ref{fig:comparing} compares the performance of indices pruned via \rpt\ against those rebuilt from scratch using identical $\alpha$ settings. While Theorem~\ref{thm:sorted} indicates that pruned and rebuilt indices possess different theoretical worst-case guarantees, this comparison serves to evaluate how effectively pruning approximates the structure of a graph constructed directly with the target sparsity. 

Perhaps surprisingly, even if an initially $\alpha_1$-reachable pruned by $\rpt$ with $\alpha_2$ has less-than-$\alpha_2$ reachability or even vacuous reachability (i.e., the reachability could be less than $1$ in the worst case) in the worst case as shown in the last section, for Vamana heuristics and real-world datasets, we see that, across all four datasets, indices pruned via \rpt\ consistently exhibit a superior QPS-recall trade-off compared to their rebuilt counterparts at the same $\alpha$. This performance gap stems from the initialization quality: \rpt\ operates on a base graph ($\alpha=1.2$) characterized by high connectivity and reachability, preserving more critical edges. Conversely, rebuilding from scratch with a small $\alpha$ constrains the graph's connectivity and reachability throughout the construction process, hindering the identification of optimal edges necessary for efficient navigation. 

The structural deficiency of the rebuilt indices is particularly evident in SIFT1M for $\alpha = 1.01, 1.05, 1.10$. In these cases, we observe a counter-intuitive phenomenon where increasing the search effort (beam size) reduces both QPS and recall simultaneously. This indicates that due to poor graph quality, a larger search beam size causes the query traversal to be misguided into local optima or irrelevant subgraphs rather than converging on the true nearest neighbors. 

\section{Related Work}

Approximate nearest-neighbor index structures are core to every vector database system, and their importance has motivated recent NeurIPS challenges~\cite{simhadri2022results,simhadri_results_2024}.  Early real-world systems leveraged the FAISS library~\cite{douze_faiss_2024}, often using its implementation of the Hierarchical Navigable Small Worlds~\cite{malkov_efficient_2018} (HNSW) algorithm. However, few theoretical properties are known about HNSW, and it targets scenarios where the entire index fits into RAM.  Other techniques more appropriate to disk-based storage, such as ANNOY~\cite{noauthor_spotifyannoy_2026} (based on random hyperplane splits into memory-mapped files), and inverted-list-based methods, such as SPANN~\cite{chen_spann_2021}, are not well understood theoretically. DiskANN~\cite{10.1145/3543507.3583552,NEURIPS2019_09853c7f, adams2025distributedannefficientscalingsingle} remains a commonly used, state-of-the-art method for disk-based indexing. Its graph-based approach has been studied theoretically~\cite{indyk2023worstcaseperformancepopularapproximate,gollapudi2025sort}.  Our work builds upon the formal understanding of DiskANN, but focuses on leveraging the theoretical insights to obtain practical index tuning.

\bibliographystyle{plainnat}
\bibliography{main}

@inproceedings{NEURIPS2019_09853c7f,
 author = {Jayaram Subramanya, Suhas and Devvrit, Fnu and Simhadri, Harsha Vardhan and Krishnawamy, Ravishankar and Kadekodi, Rohan},
 booktitle = {Advances in Neural Information Processing Systems},
 editor = {H. Wallach and H. Larochelle and A. Beygelzimer and F. d\textquotesingle Alch\'{e}-Buc and E. Fox and R. Garnett},
 pages = {},
 publisher = {Curran Associates, Inc.},
 title = {DiskANN: Fast Accurate Billion-point Nearest Neighbor Search on a Single Node},
 url = {https://proceedings.neurips.cc/paper_files/paper/2019/file/09853c7fb1d3f8ee67a61b6bf4a7f8e6-Paper.pdf},
 volume = {32},
 year = {2019}
}

@misc{indyk2023worstcaseperformancepopularapproximate,
      title={Worst-case Performance of Popular Approximate Nearest Neighbor Search Implementations: Guarantees and Limitations}, 
      author={Piotr Indyk and Haike Xu},
      year={2023},
      eprint={2310.19126},
      archivePrefix={arXiv},
      primaryClass={cs.DS},
      url={https://arxiv.org/abs/2310.19126}, 
}

@inproceedings{
gollapudi2025sort,
title={Sort Before You Prune: Improved Worst-Case Guarantees of the Disk{ANN} Family of Graphs},
author={Siddharth Gollapudi and Ravishankar Krishnaswamy and Kirankumar Shiragur and Harsh Wardhan},
booktitle={Forty-second International Conference on Machine Learning},
year={2025},
url={https://openreview.net/forum?id=JnXbUKtLzz}
}

@article{jegou2011product,
  title={Product quantization for nearest neighbor search},
  author={J{\'e}gou, Herv{\'e} and Douze, Matthijs and Schmid, Cordelia},
  journal={IEEE Transactions on Pattern Analysis and Machine Intelligence},
  volume={33},
  number={1},
  pages={117--128},
  year={2011},
  publisher={IEEE}
}

@inproceedings{babenko2016efficient,
  title={Efficient indexing of billion-scale datasets of deep descriptors},
  author={Babenko, Artem and Lempitsky, Victor},
  booktitle={Proceedings of the IEEE Conference on Computer Vision and Pattern Recognition},
  pages={2055--2063},
  year={2016}
}

@inproceedings{simhadri2022results,
  title={Results of the {NeurIPS}'21 Challenge on Billion-Scale Approximate Nearest Neighbor Search},
  author={Simhadri, Harsha Vardhan and Williams, George and Aum{\"u}ller, Martin and Douze, Matthijs and Babenko, Artem and Baranchuk, Dmitry and Chen, Qi and Hosseini, Lucas and Krishnaswamy, Ravishankar and Srinivasa, Gopal and others},
  booktitle={Proceedings of the NeurIPS 2021 Competitions and Demonstrations Track},
  volume={176},
  pages={177--189},
  year={2022},
  organization={PMLR}
}

@article{malkov_efficient_2018,
    title = {Efficient and robust approximate nearest neighbor search using hierarchical navigable small world graphs},
    volume = {42},
    number = {4},
    journal = {IEEE transactions on pattern analysis and machine intelligence},
    author = {Malkov, Yu A and Yashunin, Dmitry A},
    year = {2018},
    pages = {824--836},
}

@article{gao_retrieval-augmented_2023,
    title = {Retrieval-{Augmented} {Generation} for {Large} {Language} {Models}: {A} {Survey}},
    volume = {abs/2312.10997},
    url = {https://api.semanticscholar.org/CorpusID:266359151},
    journal = {ArXiv},
    author = {Gao, Yunfan and Xiong, Yun and Gao, Xinyu and Jia, Kangxiang and Pan, Jinliu and Bi, Yuxi and Dai, Yi and Sun, Jiawei and Guo, Qianyu and Wang, Meng and Wang, Haofen},
    year = {2023},
}

@article{zeakis_pre-trained_2023,
    title = {Pre-trained embeddings for entity resolution: an experimental analysis},
    volume = {16},
    number = {9},
    journal = {Proceedings of the VLDB Endowment},
    author = {Zeakis, Alexandros and Papadakis, George and Skoutas, Dimitrios and Koubarakis, Manolis},
    year = {2023},
    pages = {2225--2238},
}

@inproceedings{wang_milvus_2021,
    address = {Virtual Event China},
    title = {Milvus: {A} {Purpose}-{Built} {Vector} {Data} {Management} {System}},
    isbn = {9781450383431},
    shorttitle = {Milvus},
    url = {https://dl.acm.org/doi/10.1145/3448016.3457550},
    doi = {10.1145/3448016.3457550},
    language = {en},
    urldate = {2026-01-24},
    booktitle = {Proceedings of the 2021 {International} {Conference} on {Management} of {Data}},
    publisher = {ACM},
    author = {Wang, Jianguo and Yi, Xiaomeng and Guo, Rentong and Jin, Hai and Xu, Peng and Li, Shengjun and Wang, Xiangyu and Guo, Xiangzhou and Li, Chengming and Xu, Xiaohai and Yu, Kun and Yuan, Yuxing and Zou, Yinghao and Long, Jiquan and Cai, Yudong and Li, Zhenxiang and Zhang, Zhifeng and Mo, Yihua and Gu, Jun and Jiang, Ruiyi and Wei, Yi and Xie, Charles},
    month = jun,
    year = {2021},
    pages = {2614--2627},
}

@misc{upreti_cost-effective_2025,
    title = {Cost-{Effective}, {Low} {Latency} {Vector} {Search} with {Azure} {Cosmos} {DB}},
    url = {http://arxiv.org/abs/2505.05885},
    doi = {10.48550/arXiv.2505.05885},
    urldate = {2026-01-24},
    publisher = {arXiv},
    author = {Upreti, Nitish and Simhadri, Harsha Vardhan and Sundar, Hari Sudan and Sundaram, Krishnan and Boshra, Samer and Perumalswamy, Balachandar and Atri, Shivam and Chisholm, Martin and Singh, Revti Raman and Yang, Greg and Hass, Tamara and Dudhey, Nitesh and Pattipaka, Subramanyam and Hildebrand, Mark and Manohar, Magdalen and Moffitt, Jack and Xu, Haiyang and Datha, Naren and Gupta, Suryansh and Krishnaswamy, Ravishankar and Gupta, Prashant and Sahu, Abhishek and Varada, Hemeswari and Barthwal, Sudhanshu and Mor, Ritika and Codella, James and Cooper, Shaun and Pilch, Kevin and Moreno, Simon and Kataria, Aayush and Kulkarni, Santosh and Deshpande, Neil and Sagare, Amar and Billa, Dinesh and Fu, Zishan and Vishal, Vipul},
    month = jul,
    year = {2025},
    note = {arXiv:2505.05885},
}

@inproceedings{10.1145/3543507.3583552,
author = {Gollapudi, Siddharth and Karia, Neel and Sivashankar, Varun and Krishnaswamy, Ravishankar and Begwani, Nikit and Raz, Swapnil and Lin, Yiyong and Zhang, Yin and Mahapatro, Neelam and Srinivasan, Premkumar and Singh, Amit and Simhadri, Harsha Vardhan},
title = {Filtered-DiskANN: Graph Algorithms for Approximate Nearest Neighbor Search with Filters},
year = {2023},
isbn = {9781450394161},
publisher = {Association for Computing Machinery},
address = {New York, NY, USA},
url = {https://doi.org/10.1145/3543507.3583552},
doi = {10.1145/3543507.3583552},
booktitle = {Proceedings of the ACM Web Conference 2023},
pages = {3406–3416},
numpages = {11},
location = {Austin, TX, USA},
series = {WWW '23}
}

@misc{adams2025distributedannefficientscalingsingle,
      title={DISTRIBUTEDANN: Efficient Scaling of a Single DISKANN Graph Across Thousands of Computers}, 
      author={Philip Adams and Menghao Li and Shi Zhang and Li Tan and Qi Chen and Mingqin Li and Zengzhong Li and Knut Risvik and Harsha Vardhan Simhadri},
      year={2025},
      eprint={2509.06046},
      archivePrefix={arXiv},
      primaryClass={cs.DC},
      url={https://arxiv.org/abs/2509.06046}, 
}

@misc{jaiswal2022ooddiskannefficientscalablegraph,
      title={OOD-DiskANN: Efficient and Scalable Graph ANNS for Out-of-Distribution Queries}, 
      author={Shikhar Jaiswal and Ravishankar Krishnaswamy and Ankit Garg and Harsha Vardhan Simhadri and Sheshansh Agrawal},
      year={2022},
      eprint={2211.12850},
      archivePrefix={arXiv},
      primaryClass={cs.LG},
      url={https://arxiv.org/abs/2211.12850}, 
}

@article{douze_faiss_2024,
    title = {The {FAISS} library},
    journal = {arXiv preprint arXiv:2401.08281},
    author = {Douze, Matthijs and Guzhva, Alexandr and Deng, Chengqi and Johnson, Jeff and Szilvasy, Gergely and Mazaré, Pierre-Emmanuel and Lomeli, Maria and Hosseini, Lucas and Jégou, Hervé},
    year = {2024},
}

@misc{simhadri_results_2024,
    title = {Results of the {Big} {ANN}: {NeurIPS}'23 competition},
    shorttitle = {Results of the {Big} {ANN}},
    url = {http://arxiv.org/abs/2409.17424},
    doi = {10.48550/arXiv.2409.17424},
    urldate = {2026-01-29},
    publisher = {arXiv},
    author = {Simhadri, Harsha Vardhan and Aumüller, Martin and Ingber, Amir and Douze, Matthijs and Williams, George and Manohar, Magdalen Dobson and Baranchuk, Dmitry and Liberty, Edo and Liu, Frank and Landrum, Ben and Karjikar, Mazin and Dhulipala, Laxman and Chen, Meng and Chen, Yue and Ma, Rui and Zhang, Kai and Cai, Yuzheng and Shi, Jiayang and Chen, Yizhuo and Zheng, Weiguo and Wan, Zihao and Yin, Jie and Huang, Ben},
    month = sep,
    year = {2024},
    note = {arXiv:2409.17424},
}

@misc{noauthor_spotifyannoy_2026,
    title = {spotify/annoy},
    copyright = {Apache-2.0},
    url = {https://github.com/spotify/annoy},
    urldate = {2026-01-29},
    publisher = {Spotify},
    month = jan,
    year = {2026},
    note = {original-date: 2013-04-01T20:29:40Z},
}

@inproceedings{chen_spann_2021,
    title = {{SPANN}: highly-efficient billion-scale approximate nearest neighbor search},
    isbn = {9781713845393},
    booktitle = {Proceedings of the 35th {International} {Conference} on {Neural} {Information} {Processing} {Systems}},
    publisher = {Curran Associates Inc.},
    author = {Chen, Qi and Zhao, Bing and Wang, Haidong and Li, Mingqin and Liu, Chuanjie and Li, Zengzhong and Yang, Mao and Wang, Jingdong},
    year = {2021},
    pages = {Article 398},
}

\newpage
\appendix

\section{Proof of \Cref{lem:sorted-euclid-opt}}
For convenience, we restate the optimization problem here:
\begin{equation*}
\begin{aligned}
\max_{x,y,z\in\mathbb{R}^d}\quad & \|x-z\|_2 \\
\text{s.t.}\quad
& \|y-z\|_2 \le \frac{\|z\|_2}{\alpha_1},\\
& \|x-y\|_2 \le \frac{\|y\|_2}{\alpha_2},\\
& \|x\|_2 \le \|y\|_2 \le \|z\|_2 = 1
\end{aligned}
\end{equation*}

Let $\opt$ denote the optimal value of the stated problem, and define
\[
\beta(\alpha_1,\alpha_2)
:=\frac{1}{\alpha_1}\sqrt{1-\frac{1}{4\alpha_2^2}}+\frac{1}{\alpha_2}\sqrt{1-\frac{1}{4\alpha_1^2}}.
\]
We prove $\opt=\beta(\alpha_1,\alpha_2)$.

For convenience of notation, we will let $\|v\|$ denote the $\ell_2$-norm of $v$ in this proof.

\paragraph{Upper bound.} Let $f(x,y,z) = \|x-z\|^2$. By monotonicity of $\sqrt{\cdot}$, it suffices to upper bound the optimum of $f$ by $\beta(\alpha_1,\alpha_2)^2$. We relax the constraint $\|z\|=1$ to $\|z\|\le 1$ and rewrite all constraints as quadratic:\begin{equation*}
    \begin{aligned}
    & \underset{x, y, z\in\R^d}{\max}
    & & f(x,y,z) := \|x-z\|^2 \\
    & \text{subject to}
    & & g_1(x,y,z) := \|x\|^2 - \|y\|^2 \le 0 \\
    & & & g_2(x,y,z) := \|y\|^2 - \|z\|^2 \le 0 \\
    & & & g_3(x,y,z) := \|y-z\|^2 - \|z\|^2 / \alpha_1^2 \le 0 \\
    & & & g_4(x,y,z) := \|x-y\|^2 - \|y\|^2 / \alpha_2^2 \le 0 \\
    & & & g_5(x,y,z) := \|z\|^2 - 1 \le 0 \\
    \end{aligned}
    \end{equation*}

    For multipliers $\lambda_1,\ldots,\lambda_5\ge 0$, define the Lagrangian
    \[
    L(x,y,z;\lambda):=f(x,y,z)-\sum_{i=1}^5 \lambda_i g_i(x,y,z).
    \]
    Then, for any feasible $x,y,z\in\R^d$, one has $L(x,y,z;\,\lambda) \geq f(x,y,z)$, meaning the optimum value of $f$ is at most \[\sup_{x,y,z\in\R^d} L(x,y,z;\,\lambda).\]

    We will now specify our settings of $\lambda_1,\ldots,\lambda_5$ as follows. For $i\in\{1,2\}$, set \[\theta_i := \sin^{-1}\del{\frac{1}{2\alpha_i}}\in(0,\pi/6], \quad\quad h_i := \sin(2\theta_i), \quad\quad h_{12} := \sin(2(\theta_1+\theta_2)). \] Note that $1/\alpha_i^2 = 4\sin^2\theta_i = 2(1-\cos 2\theta_i)$, and also $\beta(\alpha_1,\alpha_2) = 2\sin(\theta_1+\theta_2)$. Then, define \[\lambda_1 = \frac{h_1+h_2-h_{12}}{h_2} \quad\quad \lambda_2 = \frac{h_1+h_2-h_{12}}{h_1} \quad\quad \lambda_3 = \frac{h_{12}}{h_1}\quad\quad \lambda_4 = \frac{h_{12}}{h_2}\quad\quad \lambda_5 = \beta(\alpha_1,\alpha_2)^2\]

    These are all nonnegative since $h_1,h_2,h_{12}>0$ and
$h_{12}=\sin(2\theta_1+2\theta_2)\le \sin(2\theta_1)+\sin(2\theta_2)=h_1+h_2$.

To show that $L(x,y,z;\,\lambda)\leq \beta(\alpha_1,\alpha_2)^2$, we claim that \begin{equation*}
  L(x,y,z;\,\lambda) = \beta(\alpha_1,\alpha_2)^2 - \|ax+by+cz\|^2   \tag{$\ast$} \label{eq:duality-eq}
\end{equation*}

    for appropriately chosen scalars $a,b,c\in\R$: \[a = \sqrt{h_1/h_2}\quad\quad c = \sqrt{h_2/h_1}\quad\quad b=-\frac{h_{12}}{\sqrt{h_1h_2}}.\]

It remains to verify \eqref{eq:duality-eq} by matching coefficients.
Expand
\[
\|ax+by+cz\|^2
=a^2\|x\|^2+b^2\|y\|^2+c^2\|z\|^2
+2ab\ip{x}{y}+2ac\ip{x}{z}+2bc\ip{y}{z}.
\]
\begin{itemize}
\item Constant term: $L$ has constant term $\lambda_5=\beta(\alpha_1,\alpha_2)^2$, matching the RHS.
\item Coefficient of $\ip{x}{y}$: $L$ has coefficient $2\lambda_4$, and the RHS has coefficient $-2ab$.
With $a=\sqrt{h_1/h_2}$ and $b=-h_{12}/\sqrt{h_1h_2}$, we get $-2ab=2h_{12}/h_2=2\lambda_4$.
\item Coefficient of $\ip{y}{z}$: $L$ has coefficient $2\lambda_3$, and the RHS has coefficient $-2bc$.
With $c=\sqrt{h_2/h_1}$, we get $-2bc=2h_{12}/h_1=2\lambda_3$.
\item Coefficient of $\ip{x}{z}$: $L$ has coefficient $-2$, and the RHS has coefficient $-2ac$.
Since $ac=\sqrt{h_1/h_2}\sqrt{h_2/h_1}=1$, these match.
\item Coefficient of $\|x\|_2^2$: $L$ has coefficient $1-\lambda_1-\lambda_4$.
Using $\lambda_1+\lambda_4=\frac{h_1+h_2}{h_2}$, we have
$1-\lambda_1-\lambda_4=1-\frac{h_1+h_2}{h_2}=-\frac{h_1}{h_2}=-a^2$.
\item Coefficient of $\|y\|^2$:
we must show
\[
\lambda_1-\lambda_2-\lambda_3-\lambda_4+\lambda_4/\alpha_2^2=-\frac{h_{12}^2}{h_1h_2}.
\]
Substituting the definitions of $\lambda_i$ and using $1/\alpha_2^2=2(1-\cos(2\theta_2))$ gives
\begin{align*}
\lambda_1-\lambda_2-\lambda_3-\lambda_4+\lambda_4/\alpha_2^2
&=\frac{h_1+h_2-h_{12}}{h_2}-\frac{h_1+h_2-h_{12}}{h_1}-\frac{h_{12}}{h_1}-\frac{h_{12}}{h_2}+\frac{h_{12}}{h_2\alpha_2^2}\\
&=\frac{h_1+h_2-2h_{12}+h_{12}/\alpha_2^2}{h_2}-\frac{h_2}{h_1}\\
&=\frac{h_1+h_2-2h_{12}\cos(2\theta_2)}{h_2}-\frac{h_2}{h_1}.
\end{align*}
Thus it suffices to show
\[
h_1^2+h_2^2-2h_1h_{12}\cos(2\theta_2)=h_1h_2+h_{12}^2.
\]
Equivalently,
\[
h_2^2=h_1^2+h_{12}^2-2h_1h_{12}\cos(2\theta_2),
\]
which follows from the identity $h_{12}=h_1\cos(2\theta_2)+h_2\cos(2\theta_1)$:
\begin{align*}
h_2^2
&=h_2^2\cos^2(2\theta_1)+h_2^2\sin^2(2\theta_1)\\
&=(h_{12}-h_1\cos(2\theta_2))^2+h_2^2\sin^2(2\theta_1)\\
&=h_{12}^2+h_1^2\cos^2(2\theta_2)-2h_1h_{12}\cos(2\theta_2)+h_2^2\sin^2(2\theta_1)\\
&=h_{12}^2+h_1^2-2h_1h_{12}\cos(2\theta_2),
\end{align*}
as claimed. Therefore the coefficient of $\|y\|^2$ matches $-b^2=-h_{12}^2/(h_1h_2)$.

\item Coefficient of $\|z\|^2$:
we show that
\[
1+\lambda_2-\lambda_3+\lambda_3/\alpha_1^2-\lambda_5=-\frac{h_2}{h_1}.
\]
Substituting $\lambda_2=\frac{h_1+h_2-h_{12}}{h_1}$ and $\lambda_3=\frac{h_{12}}{h_1}$ and using $1/\alpha_1^2=2(1-\cos(2\theta_1))$ yields
\begin{align*}
1+\lambda_2-\lambda_3+\lambda_3/\alpha_1^2-\lambda_5
&=1+\frac{h_1+h_2-h_{12}}{h_1}-\frac{h_{12}}{h_1}+\frac{h_{12}}{h_1\alpha_1^2}-\lambda_5\\
&=2+\frac{h_2-2h_{12}\cos(2\theta_1)}{h_1}-\lambda_5.
\end{align*}
It therefore suffices to check that
\[
\lambda_5=\beta(\alpha_1,\alpha_2)^2
=2+\frac{2h_2-2h_{12}\cos(2\theta_1)}{h_1}.
\]
Using $h_{12}=\sin(2(\theta_1+\theta_2))$ and $h_1=\sin(2\theta_1)$, $h_2=\sin(2\theta_2)$, we compute
\begin{align*}
2+\frac{2h_2-2h_{12}\cos(2\theta_1)}{h_1}
&=2+\frac{2\sin(2\theta_2)-2\sin(2\theta_1+2\theta_2)\cos(2\theta_1)}{\sin(2\theta_1)}\\
&=2-2\cos(2\theta_1+2\theta_2)\\
&=(2\sin(\theta_1+\theta_2))^2\\
&=\beta(\alpha_1,\alpha_2)^2
=\lambda_5,
\end{align*}
where the second equality uses $\sin(2\theta_1+2\theta_2)=\sin(2\theta_1)\cos(2\theta_2)+\cos(2\theta_1)\sin(2\theta_2)$.
Substituting back gives $1+\lambda_2-\lambda_3+\lambda_3/\alpha_1^2-\lambda_5=-h_2/h_1=-c^2$.
\end{itemize}

This completes the proof that $\opt \le \beta(\alpha_1,\alpha_2)$.

\paragraph{Lower bound.}
We exhibit a solution $x,y,z\in\R^d$ achieving value $\beta(\alpha_1,\alpha_2)$.
Let $z=(1,0)$, let $y=(\cos(2\theta_1),\sin(2\theta_1))$, and let
$x=(\cos(2(\theta_1+\theta_2)),\sin(2(\theta_1+\theta_2)))$,
where $\theta_i=\sin^{-1}(1/(2\alpha_i))$ as above. Then
$\|x\|_2=\|y\|_2=\|z\|_2=1$, so $\|x\|_2\le \|y\|_2\le \|z\|_2=1$.
Moreover,
\[
\|y-z\|_2=2\sin(\theta_1)=\frac{1}{\alpha_1},\qquad
\|x-y\|_2=2\sin(\theta_2)=\frac{1}{\alpha_2},
\]
so all constraints are satisfied with equality. Finally,
\[
\|x-z\|_2=2\sin(\theta_1+\theta_2)=\beta(\alpha_1,\alpha_2),
\]
so $\opt\ge \beta(\alpha_1,\alpha_2)$. 

We conclude that $\opt = \beta(\alpha_1,\alpha_2)$.

\end{document}